\documentclass[aps,pra,epsfig,showpacs,superscriptaddress,footnote]{revtex4}
\usepackage{amsmath,amsthm,amssymb}
\usepackage{graphicx}
\makeatletter
\def\be{\begin{equation}}
\def\ee{\end{equation}}
\def\bea{\begin{eqnarray}}
\def\eea{\end{eqnarray}}
\def\ben{\begin{equation*}}
\def\een{\end{equation*}}
\def\bean{\begin{eqnarray*}}
\def\eean{\end{eqnarray*}}
\def\bma{\begin{mathletters}}
\def\ema{\end{mathletters}}
\def\bi{\begin{itemize}}
\def\ei{\end{itemize}}
\newtheorem{thm}{Theorem}
\newtheorem{cor}[thm]{Corollary}
\newtheorem{lem}{Lemma}[thm]

\newcommand{\ket}[1]{ | \, #1 \rangle}

\begin{document}

\title{Quantum secret sharing based on local distinguishability}
\author{Ramij Rahaman}
\email{ramijrahaman@gmail.com}\affiliation{Department of Mathematics, University of Allahabad, Allahabad-211002, U.P., India} \affiliation{Institute of
Theoretical Physics \& Astrophysics, University of Gda\'{n}sk,
80-952 Gda\'{n}sk, Poland}
\affiliation{Department of Informatics, University of Bergen, Post Box-7803, 5020, Bergen,
Norway}

\author{Matthew G. Parker}
\affiliation{Department of Informatics, University of Bergen, Post Box-7803, 5020, Bergen,
Norway}

\begin{abstract}
In this paper we analyze the (im)possibility of the exact distinguishability of orthogonal multipartite entangled states
under {\em restricted local operation and classical communication}. Based on this local distinguishability analysis we propose a new scheme for quantum secret sharing (QSS). Our QSS scheme is quite general and cost efficient compared to other schemes. In our scheme no joint quantum operation is needed to reconstruct the secret. We also present an interesting $(2,n)$-threshold QSS scheme, where any two cooperating players, one from each of two disjoint groups of players, can always reconstruct the secret. This QSS scheme is quite uncommon, as most $(k,n)$-threshold schemes have the restriction $k\geq\lceil\frac{n}{2}\rceil$.
\end{abstract}
\maketitle

\section{Introduction}

Classical secret sharing (CSS) is a cryptographic protocol in which a dealer splits the secret and distributes shares to various shareholders. The secret can be reconstructed only when a sufficient number of shareholders cooperate with each other by sharing their individual parts of the secret. 
The CSS scheme was invented independently by Adi Shamir \cite{Sha79} and George Blakley \cite{Bla79} in 1979. The main drawback with all existing CSS schemes is that they are not perfectly secure from eavesdropper attack. To defeat the eavesdropper attack perfectly, Hillery {\em et.al.} \cite{HBB99}, and Cleve {\em et.al.} \cite{CGL99}, simultaneously presented the quantum secret sharing (QSS) scheme in 1999, whereas \.{Z}ukowski et.al. previously mentioned this type of scheme in a different context \cite{ZZHW98}. In a $(k,n)$-threshold QSS scheme, the dealer encodes her secret within a multipartite quantum state then distributes the shares of the quantum state to n pairwise distant parties. Any group of k or more parties can collaboratively reconstruct the secret from their shares while no group of fewer than k parties can. \\

Several works have considered the area of QSS \cite{Got00,NMI01,SS05,ZLM05,ZM05}. Although these QSS schemes are highly secure, either they are very restricted or they require more than one particle measurement {\em i.e.}, joint quantum measurement to reveal the secret. To perform a joint quantum operation one must bring the relevant subsystems into one place, which is, practically, quite expensive due to high quantum noise for a large number of players. Recently, Fortescue and Gour \cite{FG12} proposed a QSS scheme where they partially reduced the required quantum communication at the cost of some classical communication. But their scheme still requires some large joint quantum operations to reveal the secret. To deal with this joint measurement problem completely, we here demonstrate a simple but very efficient construction of a perfect $(k,n)$-threshold QSS scheme based on only local quantum operations and classical communication. We encode our secret into several pairs of locally distinguishable orthogonal multipartite entangled states and distribute different particles to different parties depending on the context. Only a sufficient number of cooperating parties can distinguish these pairs of orthogonal entangled states by local operation and thereby reconstruct the secret as well. Another restriction for most $(k,n)$-threshold QSS schemes is that they work only if $k\geq\lceil\frac{n}{2}\rceil$ \cite{CGL99,SS05}. In this context we provide a restricted $(2,n)$ threshold QSS scheme, where any two cooperating parties from two disjoint groups, taken over all parties, can always reconstruct the secret. 
 This type of feature is quite uncommon and not present in the existing QSS schemes.\\

Our QSS scheme is largely based on local distinguishability of a pair of orthogonal entangled states so, before describing our QSS scheme, we first discuss several local distinguishability problems that are relevant for our QSS scheme. Then we present our QSS scheme for various thresholds and also discuss how our scheme is perfectly secure from eavesdropper attack. We end with a conclusion.

\section{Local distinguishability of symmetric states}
The question of local discrimination of a set of multipartite orthogonal states is: a number of parties share a quantum system prepared in one of a known set of mutually orthogonal quantum states. The goal of all the parties is to determine the state in which the quantum system was prepared using local operations and classical communication (LOCC) \cite{WSHV00,GKRSS01,BGK11}. We now discuss several distinguishability problems related to an orthogonal pair of $n$-qubit symmetric states {\it i.e.}, GHZ and Dicke states under {\em restricted local operation and classical communication (rLOCC)}. Here, `{\em rLOCC}' means only a subset of the parties are allowed to communicate with each other.

\subsection{Local distinguishability of GHZ states}
The canonical representation of an $n(\geq 2)$-qubit GHZ state can be
written as (up to a global phase and proper basis transformation):
\bean
\label{nGHZs}
\ket{GHZ(i_1,i_2,i_3,.....,i_n)}_{123....n} =\frac{1}{\sqrt{2}}\left[\ket{0~i_2~i_3~......i_n} ~+~
(-1)^{i_1} \ket{1~\overline{i_2}~ \overline{i_3}.......\overline{i_3}}\right],
\eean

where $i_1, i_2, i_3,.....i_n \in \{ 0, 1\}$, and
the bar over a bit value indicates its logical negation. For $n=2$, these states are generally known as `\emph{Bell states}' or two qubit maximally entangled states.\\

Let's define a pair of orthogonal distance-$r (0\leq r\leq\lfloor\frac{n}{2}\rfloor)$, $n$-qubit GHZ states as follows:
\bea\label{GHZr}
\ket{GHZ}&=&\frac{1}{\sqrt{2}}\left[\ket{\overset{r}{\overbrace{000.....00}}~\overset{n-r}{\overbrace{000.....00}}}+\ket{111.....11~111.....11}\right],\nonumber\\
\ket{GHZ}_r&=&\frac{1}{\sqrt{2}}\left[\ket{\overset{r}{\overbrace{111.....11}}~\overset{n-r}{\overbrace{000.....00}}}-\ket{000.....00~111.....11}\right],
\eea
where r is a non-negative integer. If $r>\frac{n}{2}$, then we consider $n-r$ as the distance for the above pair (\ref{GHZr}).\\

\begin{thm}\label{dis-0}
No orthogonal pair of distance-$0$ $n$-qubit GHZ states can be distinguished by any less than n cooperating parties.
\end{thm}

\begin{lem}\label{bellDis}
Classical communication is necessary to distinguish any pair of Bell states locally.
\end{lem}
\begin{proof}
Without any loss of generality, let Alice and Bob share the following pair of Bell states,
\bea\label{bellpair}
\ket{\Phi^{\pm}}&=&\frac{1}{\sqrt{2}}[\ket{0}_A\ket{0}_B \pm \ket{1}_A\ket{1}_B].
\eea

Their goal is to distinguish the above pair of Bell states by only local operations (LO)
on their respective qubits and they are not allowed to communicate, classically, with each other. \\

Let Alice and Bob be spatially separated, and share the known Bell state $\ket{\Phi^{+}}$.
Bob applies $\mathbb{I}$ or $\mathbb{\sigma}_z$ on his qubit to communicate the message $0$ or $1$ respectively, and the desired state may change to an another orthogonal Bell state as,
\bea
(\mathbb{I}^A\otimes \mathbb{I}^B)\ket{\Phi^{+}}&=&\frac{1}{\sqrt{2}}[\ket{0}_A\ket{0}_B + \ket{1}_A\ket{1}_B],\nonumber\\
(\mathbb{I}^A\otimes {\mathbb{\sigma}_z}^B)\ket{\Phi^{+}}&=&\frac{1}{\sqrt{2}}[\ket{0}_A\ket{0}_B - \ket{1}_A\ket{1}_B]=\ket{\Phi^{-}}.
\eea

If Alice (alone) is able to distinguish the above pair without any communication from Bob, she can recover Bob's message as well, which is impossible as that would imply signaling \footnote{No message can travel faster than the speed of light in a vacuum.}.\\
\end{proof}

\begin{proof} {\it{ (of Theorem \ref{dis-0}):}}
Without any loss of generality, let's assume that the orthogonal GHZ pair of distance-$0$ can be distinguished by the first $m(<n)$ cooperative parties. A necessary condition for distinguishability of a pair of $n$-qubit
states under $m$-cooperative LOCC (where $m$ parties agree to do local operations on their respective qubit and share the measurement outcomes with each other) would be that: ``{\em The
pair of states should remain distinguishable when $m$ out of $n$ qubits are
operated on jointly at one place whereas the remaining $(n-m)$ qubits are kept at a different place and no classical communication is allowed between these two places}".
Keeping the condition in mind, the first $m$ qubits are kept in lab-A and the remaining $(n-m)$ qubits are kept together in a different lab (lab-B). Under this arrangement, any given pair of kind (\ref{GHZ0}), for a proper choice of basis,
can be written in the following bipartite form:
\bea
\ket{GHZ}&=&\frac{1}{\sqrt{2}}\left[\ket{0}_A\ket{0}_B+\ket{1}_A\ket{1}_B\right],\nonumber\\
\ket{GHZ}_0&=&\frac{1}{\sqrt{2}}\left[\ket{0}_A\ket{0}_B-\ket{1}_A\ket{1}_B\right],
\eea
which is equivalent to pair (\ref{bellpair}). Therefore, using Lemma \ref{bellDis} we conclude that without any classical communication between lab-A and lab-B, local distinguishability of the above pair is impossible. Hence, no less than $n$ cooperating parties can distinguish a pair of orthogonal distance-$0$ GHZ states.\\
\end{proof}

\begin{thm}\label{dis-r}
An orthogonal pair of distance-$r(\geq 1)$ GHZ states (\ref{GHZr}) can always be exactly distinguished by any two cooperating LOCC parties, one from the first $r$ and the other from the last $n-r$ parties.
\end{thm}

\begin{proof}
The proof is very simple. Both cooperating parties (one from the first $r$ and the other from the last $n-r$ parties) measure their own qubit in the computational basis (this is $\sigma_z$) locally, and if both of them get the same result, then their shared state was $\ket{GHZ}$, otherwise the state was $\ket{GHZ}_r$.
\end{proof}

\subsection{Local distinguishability of Dicke states}
The $n$-qubit symmetric Dicke state of weight $m(1\leq m<n)$ is defined by,
\be\label{DS}
\ket{m,n}=\frac{1}{\sqrt{n\choose m}}\left[\sum_{j=1}^{n\choose m} P_j(\ket{1_11_2....1_m0_{m+1}....0_n})\right],
\ee
where $\{P_j(\ket{1_11_2....1_m0_{m+1}....0_n})\}$ is the set of all possible distinct permutations
of $m$ 1's and $n-m$ 0's \cite{LLHL09}. For $m=1$ the state
given in (\ref{DS}) is generally called an $n$-qubit W state. Let us define the $n$-qubit generalized Dicke state of weight $m(1\leq m<n)$ by
\be\label{GDS}
\ket{m,n}_G=\sum_{j=1}^{n\choose m}c_jP_j(\ket{1_11_2....1_m0_{m+1}....0_n}),
\ee
where $\sum |c_j|^2=1$. We define the weight difference between two $n$-qubit Dicke states as the distance between them. For our purpose, the distance $r$ is less or equal to $n-2$. 
An orthogonal pair of distance-$r(0\leq r \leq n-2)$ generalized Dicke states of $n$-qubit can be written (under appropriate basis transformation) as,
\bea\label{DSp}
\ket{m,n}_G&=&\sum_{j=1}^{n\choose m}c_jP_j(\ket{1_11_2....1_m0_{m+1}....0_n}),\nonumber\\
\ket{m+ r,n}_G&=&\sum_{j=1}^{n\choose {m+r}} c'_j P_j(\ket{1_11_2....1_{m+ r}0_{m+ r+1}....0_n}), \mbox{~~with~}0\leq r \leq n-2, 
\eea
where $0< (m +r) <n$ and $\sum_{j=1}^{n\choose m}|c_j|^2=\sum_{j=1}^{n\choose {m+r}} |c'_j|^2=1$. For $r=0$, $\sum_{j=1}^{n\choose m} c^*_jc'_j=0$.

\begin{thm}\label{thdic}
No less than $n-r+1$ cooperating parties can perfectly distinguish any pair of orthogonal $n$-qubit generalized Dicke states (\ref{DSp}) of distance $r(0<r\leq n-2 
)$.
\end{thm}
\begin{proof}
Without loss of generality, let us assume that the first $k (<n-r+1)$ cooperating parties can perfectly distinguish the pair of orthogonal $n$-qubit generalized Dicke states (\ref{DSp}) of distance $r(0<r\leq n-2)$. Distinguishing these states by the first $k$ cooperating parties implies distinguishing them with resect to the bipartition $(1,2,....,k)$ vs. $(k+1,k+2,....,n)$ by local operation (LO) only (without any classical communication between these two partitions). Keeping this in mind we put the first k qubits in lab-A and the remaining $(n-k)$ qubits in
lab-B. Under this arrangement, any given pair of type (\ref{DSp}), for a proper choice of
basis, reduces to the form:
\bea\label{paird}
\ket{m,n}_G&=&\sum_{i,j}a_{ij}\ket{e_i}_A\ket{e_j}_B,\nonumber\\
\ket{m+ r,n}_G&=&\sum_{i,j}a'_{ij}\ket{e_i}_A\ket{e_j}_B,
\eea
where $\sum_{i,j}|a_{ij}|^2=1=\sum_{i,j}|a'_{ij}|^2$ and $\{\ket{e_i}_{A(B)}\}$ is an orthonormal basis associated to the joint subsystem $A(B)$. For any $r>0$, $a_{ij}a'_{ij}=0~\forall~i,j$. Let us define two subspaces of the Hilbert space $\mathcal{H}_A$ (associated with the first joint-subsystem $A$), $\mathcal{S}_{m}=\{\ket{e_i}_A;\mbox{~if~}\exists~j~ s.t. ~a_{ij}\neq 0\}$ and $\mathcal{S}_{m+ r}=\{\ket{e_i}_A;\mbox{~if~}\exists~j~ s.t.  ~a'_{ij}\neq 0\}$. Proof of Theorem \ref{thdic} follows immediately from Lemma \ref{lemdic}, stated below.
\begin{lem}\label{lemdic}
The pair of orthogonal bipartite states (\ref{paird}) can be perfectly distinguished under LO only, if $\mathcal{S}_m\perp \mathcal{S}_{m+ r}$ i.e., $a_{ij}a'_{il}=0~\forall~i,j,l$.
\end{lem}
 We now show that $\mathcal{S}_m \not\perp \mathcal{S}_{m+ r}$ for the pair (\ref{paird}) if $k <n-r+1$. Let us first assume that $k \leq m$. Then the product term $\ket{e_{i^*}}_A=\ket{1_11_2...1_{k}}_A\in \mathcal{S}_m\bigcap \mathcal{S}_{m+r}$ for some $i^*\in \{i\}$. Similarly, for $m<k(<n-r+1)$, the product term $\ket{e_{j^*}}_A=\ket{1_11_2...1_{m}0_{m+1}0_{m+2}...0_k}_A\in \mathcal{S}_m$ for some $j^*\in \{i\}$. Since $k<n-r+1$, {\em i.e.}, $k-m\leq n-m-r$, therefore $\ket{e_{j^*}}_A\in\mathcal{S}_{m+r}$ as well. Hence, $\mathcal{S}_m \not\perp \mathcal{S}_{m+ r}$ if $k <n-r+1$. This completes the proof.
\end{proof}

\begin{cor}\label{cordic}
Any $(n-r+1)$ cooperating parties can always perfectly distinguish any pair of orthogonal $n$-qubit generalized Dicke states (\ref{DSp}) of distance-$r(>0)$.
\end{cor}
\begin{proof}
If the combined subsystem, $A$, mentioned in the proof of Theorem \ref{thdic}, contains $(n-r+1)$-qubits, then all the $a_{ij}$'s and $a'_{ij}$'s given in (\ref{paird}) satisfy $a_{ij}a'_{il}=0,~\forall~i,j,l$. The proof then follows immediately from Lemma \ref{lemdic}.
\end{proof}

\begin{cor}\label{cordic1}
Any pair of orthogonal $n$-qubit generalized Dicke states (\ref{DSp}) of distance-$0$ can be perfectly distinguished iff all the $n$ parties cooperate with each other.
\end{cor}
\begin{proof}
Proof of the sufficiency part follows from the result that ``{\em two orthogonal states of any quantum system, shared in any proportion between any number of separated parties, can be perfectly distinguished}", as proved in \emph{Ref.} \cite{WSHV00}. The necessary part follows from Lemma \ref{lemdic}, as in this case all the $a_{ij}$'s and $a'_{ij}$'s given in (\ref{paird}) satisfy the relation $a_{ij}a'_{ij}\neq 0,~\forall~ i,j$, for all $k<n$.
\end{proof}

\section{Quantum secret sharing (QSS) scheme}
Suppose Alice wants to share a key between $n$ separated parties. The sender is Alice and the receivers are $\mbox{Bob}_1, \mbox{Bob}_2,...,\mbox{Bob}_n$ and only $k$ or more of the receivers may cooperate to recover the key, {\em i.e.}, we have a $(k, n)$-threshold secret sharing scheme. To implement this scheme Alice does the following steps:\\

 S1. Alice first prepares a large number (say $L>n$) of states chosen randomly from a specified pair of orthogonal $n$-qubit entangled states according to her requirement. Let's denote the prepared states by $\ket{S(a,b_t)}$ to keep details of each prepared state in each {\em run} \footnote{run $t$ is associated with the prepared state $\ket{S(a,b_t)}$ at time $t$.}. Here $a$ represents the state, randomly chosen from a pair of orthogonal states, that Alice prepares at time $t (=1,2,...,L)$, where $b_t=(1_t,2_t,....,n_t)$ represents the positions of all $n$-qubits of a prepared state $\ket{S(a,b_t)}$ at time $t$. \\

 S2. In order to prevent the eavesdropper, Eve, or less than $k$ dishonest Bob$_i$s', from learning the secret, Alice now prepares, at random, a different sequence, $r_i=\Pi_i(1,2,3,....,L)$, for each Bob$_i$, and sends the $i_t$-th qubit to Bob$_i$ according to the $r_i$ sequence order, where $\Pi_i$ is an arbitrary permutation of the sequence $(1,2,3,....,L)$. Note that Alice only sends the qubits not the information about $\Pi_i$. Hence, except Alice no one has the information about $\Pi_i$. It is also interesting to note that the $l$-th qubit of $r_i$ and the $l$-th qubit from $r_j$ in general may not be associated with the same entangled state, $\ket{S(a,b_t)}$, as $r_i\neq r_j$, for $i\neq j$. After receiving their associated sequence of qubits (i.e., $r_i$ for the $i$-th party), all the receivers now share $L$ $n$-qubit entangled states $\ket{S(a,r(b_t))}$. Here $r(b_t)=(\Pi_1(t),\Pi_2(t),....,\Pi_n(t))$. Nobody except Alice has any information about $\ket{S(a,r(b_t))}$.\\

S3.
Alice now randomly selects some run, say $\{t_j\}_{j=1}^u (\subset \{1,2,...,L\})$, and also
computes $n$ arbitrarily chosen permutations, $p_i$ of $\{1,2,...,u\}$, only known to herself. She then
prepares list $C_i=\{(\sigma_i(t_{p_i(j)}),\Pi_i(t_{p_i(j)}))\}_{j=1}^u$ for Bob$_i$ (for $i=1,2,..n$) and sends it to him. It is important to mention that Alice starts to sending lists $C_i$ only if all the receivers confirm the receipt of all their $L$ qubits.
After receiving the list $C_i$, Bob$_i$ measures his $\Pi_i(t_{p_i(j)})$-th qubit in the $\sigma_i(t_{p_i(j)})$ basis and sends the measurement outcome $v_i(t_{p_i(j)})$ to Alice. The details of measurement settings $\sigma_i(t_j)$ are discussed later, case by case.\\ 

S4. By analyzing the measurement results and associated measurement settings, Alice can easily detect the eavesdropper and, if there is one, she aborts the protocol and starts again with a new set of resources. \\

S5. If no eavesdropper is detected, Alice announces to the respective parties, all qubit positions of an unmeasured state $\ket{S(a,r(b_t))}$. Alice selects this $\ket{S(a,r(b_t))}$ according to her secret $a (=0/1)$. The mapping between classical bit value and orthogonal entangled pair is fixed and is communicated, securely, from Alice to all Bobs in advance. If Alice's secret is more than one bit then she reveals the qubit positions of a sequence of unmeasured states $\ket{S(a,r(b_t))}$. \\

We now discuss, in detail, the choice of measurements (i.e. S3. and S4.) and the choice of states (S1.) for different threshold scenarios.

\subsection{The $(n,n)$ threshold QSS scheme}\label{nnthres}
S1. Alice prepares the states, each chosen at random from a pair of distance-$0$ orthogonal $n$-qubit GHZ states,
\bea\label{GHZ0}
\ket{GHZ}&=&\frac{1}{\sqrt{2}}\left[\ket{000.....00}+\ket{111.....11}\right],\nonumber\\
\ket{GHZ}_0&=&\frac{1}{\sqrt{2}}\left[\ket{000.....00}-\ket{111.....11}\right].
\eea

S3. In order to fix the choice of measurement for a selected run from $\{t_s\}_{s=1}^u$, Alice randomly chooses an $n$-tuple binary vector, $\mathcal{O}_{t_s}$, from $\{\mathcal{O}\}\equiv\{\mathcal{O}(0)=(0_10_20_30_4...0_n)\}\bigcup\{\mathcal{O}(ij)=(0_10_2...0_{i-1}1_i0_{i+1}...0_{j-1}1_j0_{j+1}...0_n);~\forall~ i\neq j\}$ and, for the $t_{p_i(s)}$-th run from the list $C_i=\{(\sigma_i(t_{p_i(s)}),\Pi_i(t_{p_i(s)}))\}_{s=1}^u$,
either selects observable $\sigma_i(t_{p_i(s)})=\sigma_x=\left(
\begin{array}{cc}
0 & 1 \\
1 & 0 \\
\end{array}
\right)$ or observable $\sigma_i(t_{p_i(s)})=\sigma_y=\left(
\begin{array}{cc}
0 & -i \\
i & 0 \\
\end{array}
\right)$, depending on whether the $i$-th bit value of the binary vector $\mathcal{O}_{t_s}$ is $0$ or $1$, respectively.
Here, $\mathcal{O}_{t_s}$ refers both to the binary vector and to its associated stabilizer. Alice now sends list $C_i$ to Bob$_i$, and Bob$_i$ measures his $\Pi_i(t_{p_i(s)})$-th qubit in the $\sigma_i(t_{p_i(s)})$ basis and sends the measurement outcome, $v_i(t_{p_i(s)})$, to Alice for all such $(\sigma_i(t_{p_i(s)}),\Pi_i(t_{p_i(s)}))\in C_i$.\\

The $\{\mathcal{O}\}$ are the stabilizers (or anti stabilizers \footnote{$\mathcal{O}$ is called an anti-stabilizer if $\mathcal{O}\ket{\Psi}=-\ket{\Psi}$}) of the pair of states given in (\ref{GHZ0}) and they satisfy the following eigenvalue relation,
\be
\mathcal{O}_t\ket{S(a,r(b_t))}=\lambda(a,t)\ket{S(a,r(b_t))},~\forall \mathcal{O}_t\in \{\mathcal{O}\},
\ee where $\lambda(a,t)$ is an eigenvalue with $\lambda(a,t) \in \{\pm 1\}$. Therefore, the product of all individual local outcomes $v_i(t)$ for the observable $\mathcal{O}_t$ must be equal to the corresponding eigenvalue for the state $\ket{S(a,r(b_s))}$ i.e., $\lambda(a,t)=\Pi_{i=1}^n v_i(t)$.\\

S4. For each selected run $t_s$ Alice checks whether the products of local results satisfies the corresponding eigenvalue equation or not:

\bean
\mathcal{O}_{t_s}\ket{GHZ}=\begin{cases}
+1\ket{GHZ}, & \mbox{if }\mathcal{O}_{t_s}=\mathcal{O}(0), \\
-1\ket{GHZ} & \mbox{otherwise},\end{cases}& \mbox{and}& \mathcal{O}_{t_s}\ket{GHZ}_0=\begin{cases}
-1\ket{GHZ}_0, & \mbox{if }\mathcal{O}_{t_s}=\mathcal{O}(0), \\
+1\ket{GHZ}_0 & \mbox{otherwise}.\end{cases}
\eean
This protocol is secure in two ways. Firstly, the eavesdropper does not have any information about the sequence of qubits, so she cannot create a measurement outcome to satisfy all the above relations. Secondly, the above relations hold specifically for unique states (up to local unitary equivalence), so the action of an eavesdropper at any stage will be detected as this uniqueness will be compromised. \\

Theorem \ref{dis-0} tell us that the pair, (\ref{GHZ0}), can only be distinguished if all parties cooperate, otherwise the secret key cannot be recovered.

\subsection{The restricted $(2,n)$ threshold QSS scheme}
S1. Alice prepares the states, each chosen at random from a pair of distance-$r$ orthogonal $n$-qubit GHZ states, as given in Eq. (\ref{GHZr}).\\

S3. Same as step S3 described in section \ref{nnthres}.\\

S4. The product of local results of each measurement satisfies the following eigenvalue relations,

\bean
\mathcal{O}_{t_s}\ket{GHZ}=\begin{cases}
+1\ket{GHZ}, \mbox{ if }\mathcal{O}_{t_s}=\mathcal{O}(0), \\
-1\ket{GHZ} \mbox{ otherwise},\end{cases}& \mbox{and}& \mathcal{O}_{t_s}\ket{GHZ}_r=\begin{cases}
+1\ket{GHZ}_r, \mbox{ if }\mathcal{O}_{t_s}=\mathcal{O}(ij) &\mbox{with $i,j\leq r$},~\forall ~i\neq j, \\
&\mbox{~or, $i,j\geq r$},~\forall ~i\neq j, \\
-1\ket{GHZ}_r \mbox{ otherwise}.&
\end{cases}
\eean
Here, also, the above relations uniquely define the states and hence the protocol is secure in the same two ways as described in section \ref{nnthres}.

According to theorem \ref{dis-r}, any two cooperating parties, one from the first $r$ parties and the other from the remaining $n-r$ parties, can distinguish the pair (\ref{GHZr}) perfectly. Hence they are also able to recover the key.

\subsection{The $(k,n)$ threshold QSS scheme}

S1. Alice prepares the states, each chosen at random from a pair of distance-$r(>0)$ orthogonal $n$-qubit Dicke states, as given in Eq. (\ref{DSp}), with $c_j = c$ and $c'_i = c'$ both fixed real constants.
Note that we already discussed the $(n,n)$-threshold QSS scheme in section \ref{nnthres}, so here we consider the case when $k<n$.\\

S3. 
All $n$-qubit Dicke states $\ket{m,n}$ given in Eq. (\ref{DS}) are eigenstates of $\sigma^{\otimes n}_z$ with eigenvalue $(-1)^m$, where
$\sigma_z=\left(
\begin{array}{cc}
1 & 0 \\
0 & -1 \\
\end{array}
\right)$. For even $n$, the Dicke state $\ket{\frac{n}{2},n}$ is also an eigenstate of $\sigma^{\otimes n}_x$ and $\sigma^{\otimes n}_y$ with eigenvalue $1$ \cite{CLSW10}. Therefore, if $n$ is even Alice chooses (for more security) the pair of states such that the pair contains the state $\ket{\frac{n}{2},n}$. \\

To detect an eavesdropper for a $(k,n)$-threshold QSS scheme with even $n$, Alice chooses (randomly) the same measurement settings [{\em i.e.} $\sigma_1(t_{p_1(s)})=\sigma_2(t_{p_2(s)})=....=\sigma_n(t_{p_n(s)})$] from $\{\sigma_x,\sigma_y,\sigma_z\}$ for all Bob$_i$s' for the run $t_s$. Whereas for odd $n$ Alice prepares $L$, $(n+1)$-qubit states $\ket{S(a,b_t)}$ instead of n-qubit states. She keeps one qubit from each of $\ket{S(a,b_t)}$ and behaves like player Bob$_{(n+1)}$ from step S2 to S4. But she does not take any part to reveal the key/secret {\em i.e.}, in step S5.\\


S4. If there is no eavesdropping then the product of all local measurement results must satisfy the following eigenvalue relations,
\bean
\sigma_z\sigma_z...\sigma_z\ket{m,n}&=&(-1)^m\ket{m,n}\\
\sigma_x\sigma_x...\sigma_x\ket{\frac{n}{2},n}&=&\ket{\frac{n}{2},n}\text{ [when $n$ is even]}\\
\sigma_y\sigma_y...\sigma_y\ket{\frac{n}{2},n}&=&\ket{\frac{n}{2},n} \text{ [when $n$ is even]}.\\
\eean
According to Theorem \ref{thdic}, to distinguish the pair (\ref{DSp}) perfectly, cooperation between $k=n-r+1$ or more parties is necessary and, thus, the secret key is only revealed if $k\geq\lceil\frac{n}{2}\rceil$ or more parties cooperate.

\section{Conclusion}

In this paper we explore various interesting cases of quantum secret sharing (QSS) schemes based on local distinguishability of orthogonal multipartite entangled states. To reconstruct the secret, a group of single qubit quantum operations are sufficient for our schemes. Therefore our schemes are cost-efficient as well as quite competent compared to other existing QSS schemes, where several multiparty quantum operations are required to reveal the secret. Multipartite joint quantum operations are fairly expensive when the individual parties are in different places and the quantum cost increases exponentially as the number of parties increase. We also demonstrate an unusual $(2,n)$-threshold scheme where the set of players are partitioned into two disjoint groups and where the key can be recovered if any two players cooperate, on condition that the two players do not belong to the same group. This setup is quite non-standard for a QSS scheme and we hope that this result will encourage researchers to develop the field further.

\section{Acknowledgments}
We thank Marek \.{Z}ukowski for stimulating discussions. This project was started when RR was a doctoral student at UiB with the financial support by Norwegian Research Council. RR also acknowledges partial support
by Foundation for Polish Science (FNP) TEAM/2011-8/9 project co-financed by EU European Regional Development Fund, and ERC grant QOLAPS (291348).

\end{document}